\documentclass[10pt]{article}
\usepackage[letterpaper,hmargin=0.972in,vmargin=1.0in]{geometry}
\usepackage{amsmath}
\usepackage{times}
\usepackage{xspace}
\usepackage{datetime}
\usepackage{url}
\usepackage{hyperref}
\usepackage[usenames,dvipsnames]{color}
\usepackage{graphicx}
\usepackage{epsfig}
\usepackage[table,xcdraw]{xcolor}
\usepackage{multirow}
\graphicspath{{figures/}}
\usepackage{amssymb}
\usepackage{algorithm}
\usepackage[noend]{algpseudocode}
\usepackage{cite}
\usepackage{caption}
\usepackage{amsthm}
\newcommand{\coloredcomment}[3]{}

\newcommand{\deadcomment}[1]{}
\usepackage{textgreek}
\usepackage{mathtools}

\newtheorem{theorem}{Theorem}[section]
\newtheorem{lemma}[theorem]{Lemma}
\newtheorem{corollary}[theorem]{Corollary}

\date{}

\author{Petros Maniatis\qquad Ilya Mironov\qquad Kunal Talwar\\
Google Brain
}

\begin{document}

\title{Oblivious Stash Shuffle}

\maketitle

\begin{abstract}
This is a companion report to Bittau et al.~\cite{StashShuffler}. We restate and prove security of the Stash Shuffle.
\end{abstract}


\section{Description of the Stash Shuffle}
\begin{algorithm}[H]
\caption{The Stash Shuffle algorithm.\label{alg:StashShuffle}}
\footnotesize
\begin{algorithmic}[1]
\Procedure{StashShuffle}{Untrusted arrays $\textit{in}, \textit{out}, \textit{mid}$}
\State $\textit{stash} \gets \phi$\label{code:shuffle:distribute:1}
\For{$j \gets 0, B-1$}
\State $\textsc{DistributeBucket}(\textit{stash}, j, \textit{in}, \textit{mid})$
\EndFor
\State $\textsc{DrainStash}(\textit{stash}, B, \textit{mid})$\label{code:shuffle:drain}
\State $\textbf{FAIL}\ \textbf{on}\ \neg \textit{stash}.\textit{Empty}()$\label{code:shuffle:distribute:2}
\State $\textsc{Compress}(\textit{mid}, \textit{out})$\label{code:shuffle:compress}
\EndProcedure
\end{algorithmic}
\end{algorithm}\vspace*{-1ex}

Algorithm~\ref{alg:StashShuffle}, the Stash Shuffle, considers input ($\textit{in}$) and output ($\textit{out}$) items in $B$ sequential buckets, each holding at most $D \triangleq \lceil N/B\rceil$ items, sized to fit in private
memory. At a high level, the algorithm first chooses a random output bucket for
each input item, and then randomly shuffles each output bucket. It does that in two phases. During the
\emph{Distribution Phase} (lines~\ref{code:shuffle:distribute:1}--\ref{code:shuffle:distribute:2}), it reads in one input bucket at a time,
splits it across output buckets, and stores the split-up but as yet unshuffled, re-encrypted items in an
intermediate array ($\textit{mid}$) in untrusted memory.  During the 
\emph{Compression Phase} (line~\ref{code:shuffle:compress}), it reads the intermediate array of encrypted
items one bucket at a time, shuffles each bucket, and stores it
fully-shuffled in the output array.

The algorithm gets its name from the \emph{stash}, a private structure, whose
purpose is to reconcile
obliviousness with
the variability in item counts distributed across the output buckets.
This variability (an inherent result of balls-and-bins properties)
must be hidden from external observers,
and not reflected in non-private memory.
For this,
the algorithm caps the number traveling from an input bucket to an output bucket at
$C \triangleq D/B+\alpha\sqrt{D/B}$ for a
small constant $\alpha$.
If any input bucket
distributes more than $C$ items to an output bucket,
overflow items are instead stored in a stash---of size $S$---where
they queue, waiting to be
drained into the chosen output bucket during processing of later input buckets.

\begin{algorithm}[H]
\caption{Distribute one input bucket.\label{alg:Distribution}}
\footnotesize
\begin{algorithmic}[1]
\Procedure{DistributeBucket}{$\textit{stash}$, $b$, Untrusted arrays $\textit{in}$, $\textit{mid}$}
\State $\textit{output} \gets \phi$
\State$\textit{targets} \gets \textsc{ShuffleToBuckets}(B, D)$\label{code:distr:targets} 
\For{$j \gets 0, B-1$}\label{code:distr:empty_stash:1}
\While{$\neg \textit{output}[j].\textit{Full}() \wedge \neg \textit{stash}[j].\textit{Empty}() $}
\State $\textit{output}[j].\textit{Push}(\textit{stash}[j].\textit{Pop}())$\label{code:distr:empty_stash:2}
\EndWhile
\EndFor
\For{$i \gets 0, D-1$}\label{code:distr:bucket:1}
\State $\textit{item} \gets \textit{Decrypt}(\textit{in}[\textit{DataIdx}(b, i)])$

\If{$\neg \textit{output}[targets[i]].\textit{Full}()$}
\State $\textit{output}[targets[i]].\textit{Push}(\textit{item})$\label{code:distr:output_chunk}
\Else
\If{$\neg \textit{stash}.\textit{Full}()$}
\State $\textit{stash}[targets[i]].\textit{Push}(\textit{item})$
\Else
\State $\textbf{FAIL}$\label{code:distr:bucket:fail}
\EndIf
\EndIf\label{code:distr:bucket:2}

\EndFor
\For{$j \gets 0, B-1$}\label{code:distr:output:1}
\While{$\neg \textit{output}[j].\textit{Full}()$}
\State $\textit{output}[j].\textit{Push}(\textit{dummy})$
\EndWhile
\For{$i \gets 0, C-1$}
\State $\textit{mid}[\textit{MidIdx}(j, i)] \gets \textit{Encrypt}(\textit{output}[j][i])$\label{code:distr:output:2}
\EndFor
\EndFor
\EndProcedure
\end{algorithmic}
\end{algorithm}\vspace*{-1ex}

Algorithm~\ref{alg:Distribution} describes the distribution in more
detail, implementing the same logic, but reducing data copies.
$\textsc{ShuffleToBuckets}$ randomly shuffles the $D$ items of the input bucket, and $B-1$ bucket separators. The
shuffle determines which item will fall into which target bucket, stored in
$\textit{targets}$ (line~\ref{code:distr:targets}).  Then, for every output
bucket, as long as there is still room in the maximum $C$ items to output, and there are stashed away items, the output takes items from
the stash
(lines~\ref{code:distr:empty_stash:1}--\ref{code:distr:empty_stash:2}). Then the
input bucket items are read in from the outside input array, decrypted, and
deposited either in the output (if there is still room in the quota $C$ of the
target bucket), or in the stash
(lines~\ref{code:distr:bucket:1}--\ref{code:distr:bucket:2}).  Finally, if some
output chunks are still not up to the $C$ quota, they are filled with dummy
items, encrypted
and written out into the intermediate array
(lines~\ref{code:distr:output:1}--\ref{code:distr:output:2}). Note that the
stash may end up with items left over after all input buckets have been
processed, so we drain those items (padding with dummies), filling $K$ extra
items per output bucket at the end of the distribution phase
(line~\ref{code:shuffle:drain} of Algorithm~\ref{alg:StashShuffle}, which is
similar to distributing a bucket, except there is no input bucket to
distribute). $K$ is set to $S/B$, that is, the size of the stash divided by the number of buckets.

\begin{algorithm}[H]
  \caption{Compress intermediate items.\\\footnotesize ($\textit{L} \triangleq \min(W, B)$ is the
  effective window size, defined to account for pathological cases where $W >
  B$.) \label{alg:Compress}}
\footnotesize
\begin{algorithmic}[1]
\Procedure{Compress}{Untrusted arrays $\textit{mid}, \textit{out}$}
\For{$b \gets 0, L-1$}
\State $\textsc{ImportIntermediate}(b, \textit{mid})$
\EndFor
\For{$b \gets L, B-1$}
\State $\textsc{DrainQueue}(b-L, \textit{mid}, \textit{out})$
\State $\textsc{ImportIntermediate}(b, \textit{mid})$
\EndFor
\For{$b \gets B-L, B-1$}
\State $\textsc{DrainQueue}(b, \textit{mid})$
\EndFor
\EndProcedure
\end{algorithmic}
\end{algorithm}\vspace*{-1ex}

Algorithm~\ref{alg:Compress} shows the compression phase.
In this phase, the intermediate items deposited by the distribution
phase must be shuffled, and dummy items must be filtered out. To do this,
without revealing information about the distribution of (real) items in output
buckets, the phase proceeds in a sliding window of $W$ buckets of intermediate
items. The window size $W$ is meant to absorb the elasticity of real item counts
in each intermediate output bucket due to the Binomial distribution. See

\begin{algorithm}
\caption{Import an intermediate bucket.\label{alg:ImportIntermediate}}
\footnotesize
\begin{algorithmic}[1]
\Procedure{ImportIntermediate}{$\textit{b}$, Untrusted array $\textit{mid}$}
\State $\textit{bucket} \gets \textit{mid}[\textit{MidIdx}(b, 0..C*B+K-1)]$\label{code:import:load}
\State $\textit{Shuffle}(bucket)$\label{code:import:shuffle}
\For{$i \gets 0, C*B+K-1$}
\State $\textit{item} \gets \textit{Decrypt}(\textit{bucket}[i])$
\If{$\neg \textit{item}.dummy$}
\State $\textit{queue}.\textit{Enqueue}(\textit{item})$
\EndIf
\EndFor
\EndProcedure
\end{algorithmic}
\end{algorithm}\vspace*{-1ex}

As Algorithm~\ref{alg:ImportIntermediate} shows, an intermediate bucket is loaded into
private memory ($C$ items per input bucket, plus another $K$ items for the final
stash drain) in line~\ref{code:import:load}, and shuffled in line~\ref{code:import:shuffle}. Then 
intermediate items are decrypted, throwing away dummies, and enqueued for export,
$D$ items at a time, into the
output array in untrusted memory. 

The distribution step is constrained by the content of a single bucket $D$ and the stash size (although the latter may be organized on a per-bucket basis, only $S/B$ items must be kept in memory at any one time). The compression step requires keeping $CB+K$ items in memory for the freshly loaded bucket, $D(W-1)$ items for the buckets previously processed, and $Q$ items as a hedge against overflow. We discuss constraints on these parameters in the next section and offer some concrete choices in Section~\ref{s:parameters}.

\section{Security Argument}

This section is dedicated to proving our main result, Theorem~\ref{th:main}, namely that the output of the Stash Shuffle is statistically close to the uniform distribution on permutations on $N$ items. The shuffle's obliviousness is established by inspection.

Our proof follows the following steps. First, we define the Buckets Shuffle---an idealized and simplified version of the Stash Shuffle, which operates over unbounded data structures in clear text and never fails. We demonstrate that the output of the Buckets Shuffle is truly uniform. Second, we argue that the output of the real algorithm deviates from the Buckets Shuffle only when it fails, thus bounding the statistical distance between the Stash Shuffle's output distribution and the uniform distribution. Finally, we bound the probability of failure of the Stash Shuffle, and select its parameters so that the probability is negligible in $N$.

\begin{algorithm}
	\caption{The Buckets Shuffle algorithm.\label{alg:BucketsShuffle}}
	\begin{algorithmic}[1]
		\Procedure{DistrBucketIdeal}{$b$, arrays $\textit{in}, \textit{mid}$}
		\State$\textit{targets} \gets \textsc{ShuffleToBuckets}(B, D)$\label{code:distr:shuffle}\Comment{Same as before}
		\For{$i \gets 0, D-1$}
		\State $\textit{item} \gets \textit{in}[\textit{DataIdx}(b, i)]$
		\State $\textit{mid}[\textit{targets}[i]].\textit{Push}(\textit{item})$
		\EndFor
		\EndProcedure
		\Statex
		\Procedure{CompressIdeal}{array $\textit{mid}$, list $\textit{out}$}
		\For{$i \gets 0, B-1$}
		\State $\textit{bucket} \gets \textit{mid}[i]$
		\State $\textit{Shuffle}(\textit{bucket})$\label{code:compress:shuffle}
		\State $\textit{out}.\textit{Append}(\textit{bucket})$
		\EndFor
		\EndProcedure
		\Statex
		\Procedure{Shuffle}{arrays $\textit{in}, \textit{out}, \textit{mid}$}
		\For{$j \gets 0, B-1$}\label{code:shuffle:for1}
		\State $\textsc{DistrBucketIdeal}(j, \textit{in}, \textit{mid})$\label{code:shuffle:for2}
		\EndFor
		\State $\textsc{CompressIdeal}(\textit{mid}, \textit{out})$
		\EndProcedure
	\end{algorithmic}
\end{algorithm}

\begin{lemma}\label{lemma:BS-uniform}The output of the Buckets Shuffle (the \textsc{Shuffle} procedure, Algorithm~\ref{alg:BucketsShuffle}) is uniform.
\end{lemma}
\begin{proof} First, observe that the shuffle, i.e., the mapping of its input to the output, is independent of the content of the $\textit{in}$ array. It means, in particular, that if the input is uniformly sampled from the set of all permutations on $N$ elements, the output will be uniformly distributed as well.

Second, we construct a coupling between the shuffle seeded with a uniformly distributed input $\textit{in}_1$ and an arbitrary $\textit{in}_2$ as follows. For each assignment of items from the $b$th bucket of $\textit{in}_1$ output by \textsc{ShuffleToBuckets} (line~\ref{code:distr:shuffle}), we force identical assignment for the same items in $\textit{in}_2$ (possibly from different input buckets). After execution of the lines~\ref{code:shuffle:for1}--\ref{code:shuffle:for2} the internal state of the two runs of the algorithm (the content of \textit{mid}) become identical, from which the claim follows.
\end{proof}

\begin{lemma}\label{lemma:SS-BS}Statistical distance between the distributions of output of the Stash Shuffle and the Buckets Shuffle is bounded by the probability that the Stash Shuffle fails.
\end{lemma}
\begin{proof} Condition on the event that the Stash Shuffle does not fail. We again proceed by a coupling argument. If the outputs of \textsc{ShuffleToBuckets} for both shuffles are identical, the assignment of items to buckets will be the same between the two shuffles. Since the stash does not overflow, and it is drained fully (line~\ref{code:shuffle:drain} of Algorithm~\ref{alg:StashShuffle}), the buckets are perfectly matched. Then, by coupling the outputs of $\textit{Shuffle}(\textit{bucket})$ steps (line~\ref{code:compress:shuffle} of Algorithm~\ref{alg:BucketsShuffle} and line~\ref{code:import:shuffle} of Algorithm~\ref{alg:ImportIntermediate}), we ensure that the outputs of the compression steps are also identical.

By the standard probability theory argument, if two distributions are identical if one of them is conditioned on a certain event not happening, the statistical distance between the two distributions is bounded by the probability of that event.

\end{proof}

The following lemmas form the technical heart of the argument. They bound the probability of each cause of the Stash Shuffle's failing to run to completion: (1) the stash's overflowing (Algorithm~\ref{alg:Distribution}, line~\ref{code:distr:bucket:fail}); (2) the stash's not draining (Algorithm~\ref{alg:StashShuffle}, line~\ref{code:shuffle:distribute:2}); and (3) the compression algorithm's queue overflowing or underflowing.

\begin{lemma}\label{lemma:stash-overflows}Let the total capacity of the stash be $S$. Then the probability that the stash overflows or it fails to drain is bounded by
\[
F_1\leq B^2  e^{(CB/D-1)(2C-S/B)},
\]
subject to additional conditions that $K\geq S/B>2C$ and $e^t<1+(tC-\ln 2)B/D$ where $t=CB/D-1$. 
\end{lemma}
\begin{proof}
Let the number of items in $\textit{stash}[j]$ before distributing the $i$th bucket be $x_i^{(j)}$, and let $X_i^{(j)}$ be its distribution (for compactness the index $j$ is dropped when it is clear from context). The probability of the stash's overflowing is
\[
F_1=1-\Pr\left[\forall 0\leq i\leq B:\sum_{j=0}^{B-1} x_i^{(j)}\leq S\right].
\]

To bound $F_1$ we observe that the distribution $X_i$ satisfies the following recurrence for all $i$:
\begin{align}
X_0&=0,\notag\\
X_{i+1}&=\max\left(0, X_i+\mathrm{Bin}(1/B,D)-C\right),\label{eq:recurrence}
\end{align}
where $\mathrm{Bin}(\cdot,\cdot)$ is the binomial distribution.

Towards bounding the tails of $X_i$, we define two moment-generating functions as 
\[
S_i(t) = E[e^{tX_i}]\textrm{ and }S_{\mathrm{Bin}}(t) = E[e^{t\cdot \mathrm{Bin}(1/B,D)}].\
\]
These two statements are implied by the recurrence~(\ref{eq:recurrence}):
\begin{align*}
\textrm{if }X_t\leq C,&\textrm{ then } S_{i+1}(t)\leq S_{i}(t)S_{\mathrm{Bin}}(t)\leq e^{tC}S_{\mathrm{Bin}}(t),\\
\textrm{if }X_t> C,&\textrm{ then } S_{i+1}(t)=e^{-tC} S_{i}(t)S_{\mathrm{Bin}}(t),
\end{align*}
which we use in the following bound:
\[
S_{i+1}(t)=E_{x_i\leftarrow X_i}\left[E\left[e^{tX_{i+1}}\mid X_{i}=x_i\right]\right]\\
\leq \left\{e^{tC}+e^{-tC} S_{i}(t)\right\}S_{\mathrm{Bin}}(t).
\]
Fix $t_0$ so that $S_\mathrm{Bin}(t_0)=.5e^{t_0C}$. By an inductive argument it follows that for all $i\geq 0$ and $t< t_0$:
\[
S_i(t)\leq S_{\mathrm{Bin}}(t)\cdot \frac{e^{tC}}{1-e^{-tC}S_{\mathrm{Bin}}(t)}<e^{2tC}.
\]
An upper bound on the moment-generating function implies a bound on the tail probability event for any threshold $\alpha>0$:
\[
\Pr[X_i>\alpha]=\Pr[e^{tX_i}>e^{t\alpha}]\leq e^{-t\alpha}S_i(t).
\]
Thus, the probability that the size of a single $\textit{stash}[j]$ exceeds $\alpha$ is capped by $e^{t(2C-\alpha)}$, which is minimized for $t=t_0$ under the condition that $\alpha>2C$.

Setting the threshold $\alpha=S/B$ and taking the union bound over $B^2$ events $X_i^{(j)}>\alpha$, we obtain the bound on~$F_1:$
\begin{align*}
F_1&\leq \Pr\left[\forall 0\leq i\leq B: x_i^{(j)}\leq S/B\right]\\
&\leq B^2  e^{t_0(2C-S/B)}.
\end{align*}
We note that, under the conditions $x_B^{(j)}\leq S/B\leq K$, all stashes drain, which takes care of the second cause of the shuffle's failure.

To finish the argument we need to compute a lower bound on $t_0$. This is done by using an explit formula for $S_\mathrm{Bin}(t) = \left[ 1 +(e^t-1)/B\right]^D<e^{D(e^t-1)/B}$. To solve $S_\mathrm{Bin}(t)<.5e^{tC}$ for $t$, we observe that this is implied by $t$ satisfying $D(e^t-1)<(tC-\ln 2)B$. The last inequality holds for $t=CB/D-1$ in the regime of interest to us (when $CB/D=1+o(1)$ and $B\ll D$).
\end{proof}

\begin{lemma}\label{lemma:queue-fails}The probability that the compression algorithm fails is bounded by
\[
F_2\leq B\cdot\left\{\exp(-2(DW)^2/N)+\exp(-2Q^2/N)\right\}.
\]
assuming that $L=W\leq B$.
\end{lemma}
\begin{proof}Consider the following shuffle, which is a hybrid between the Stash Shuffle and the Buckets Shuffle. It follows the Buckets Shuffle in the distribution stage, and switches to the Stash Shuffle for the compression step.

Concretely, the compression algorithm of the hybrid shuffle works as follows. It reads one bucket at a time, deposits $D$ elements, while keeping $WD+Q$ elements in memory to absorb variability in the fill quotient among the buckets.

Conditional on reaching the compression step, the failure probability of the hybrid shuffle and the Stash Shuffle are identical. Thus, it suffices to analyze the failure the probability of the hybrid shuffle, which we do below.

Define the total number of items in the first $i$ buckets as $Y_i$, where $Y_0=0$ and $Y_B=N$. The probability of the compression's failing is thus bounded by
\begin{align*}
F_2&=1-\Pr[\forall W\leq i\leq B\colon 0\leq Y_i-D(i-W)\leq WD+Q],\\
&\leq \sum_{i=W}^B \Pr[Y_i<D(i-W)]+\Pr[Y_i>Di+Q].
\end{align*}

We observe that $Y_i=\mathrm{Bin}(i/B, N)$ (this is where we use the fact that the hybrid's shuffle distribution stage never fails). Recall that $D=N/B$, and thus $E[Y_i]=iD$.

The tails of the binomial distribution are bounded as
\begin{align*}
\Pr\left[\mathrm{Bin}(i/B, N)<D(i-W)\right]&\leq \exp(-2(DW)^2/N),\\
\Pr[\mathrm{Bin}(i/B, N)>Di+Q]&\leq \exp(-2Q^2/N),
\end{align*}
which implies the claim.
\end{proof}

\begin{theorem}The statistical distance between the output of the Stash Shuffle and the uniform distribution is bounded by 
\[
B^2 \exp\left\{(CB/D-1)(2C-S/B)\right\}+B\cdot\left\{\exp(-2(DW)^2/N)+\exp(-2Q^2/N)\right\},
\]
assuming $K\geq S/B>2C$, $W\leq B$ and $e^t<1+(tC-\ln 2)B/D$ where $t=CB/D-1$. 

\end{theorem}\label{th:main}
\begin{proof}
Follows by combining Lemmas~\ref{lemma:BS-uniform} and~\ref{lemma:SS-BS}, and by collecting the bounds on $F_1$ and $F_2$ (Lemmas~\ref{lemma:stash-overflows} and~\ref{lemma:queue-fails}).
\end{proof}

A simplified, asymptotic statement of Theorem~\ref{th:main} is given by the next corollary. We omit computational assumptions on security of the encryption; the rest of the argument provides unconditional security.

\begin{corollary} There is an oblivious shuffle on $N$ items with $N^{1/2+o(1)}$ private memory whose output distribution is distance $\textrm{negl}(1/N)=N^{-\omega(1)}$ from the uniform.
\end{corollary}
\begin{proof}
Consider the Stash Shuffle with the following parameters: $B=N^{1/2-\epsilon}$, $D=N^{1/2+\epsilon}$, $C = (1+\epsilon)N^{2\epsilon}$, $S=N^{1/2+2\epsilon}$, $K=N^{3\epsilon}$, $W=1$, and $Q=N^{1/2+\epsilon}$. According to Theorem~\ref{th:main}, the distance between the shuffle's output distribution and the uniform is bounded by
\[
N^{1-2\epsilon}\exp\{-(1+\epsilon)N^{3\epsilon}\}+2N^{1/2-\epsilon}\exp\left\{-N^{2\epsilon}\right\}.
\]
(The theorem's assumptions hold; the only one that requires verification is that $t=CB/D-1=\epsilon$ and $e^\epsilon\leq 1+(\epsilon(1+\epsilon)N^{2\epsilon}-\ln 2)N^{-2\epsilon}$, which holds for $\epsilon<1$ and $N^{-2\epsilon}\ln 2\ll \epsilon^2$).

As long as $1\gg \epsilon\gg \ln\ln N/\ln N$, the distance will decay faster than a negligible function in $1/N$, and the private memory size will approach $N^{1/2}$.
\end{proof}

\section{Sample Parameters}\label{s:parameters}

Table~\ref{tab:stash-shuffle} lists the Stash Shuffle's parameters for several select scenarios. Rather than applying the generic bound of our main theorem, we use a tighter estimate of the shuffle's security level (equivalently, the failure probability of the Stash Shuffle, according to Lemma~\ref{lemma:SS-BS}). The more precise bounds follow by computing the tail probabilities of the distributions $X_i$ and $Y_i$ in Lemmas~\ref{lemma:stash-overflows} and~\ref{lemma:queue-fails} respectively.

\begin{table}[htbp]
  \centering
  \small
\begin{tabular}{r|r|r|r|r|r|r||r}
\multicolumn{1}{c|}{$N$} &
    \multicolumn{1}{c|}{$B$} &
	    \multicolumn{1}{c|}{$D$} &
    	    \multicolumn{1}{c|}{$C$} &
        	    \multicolumn{1}{c|}{$W$} &
            	    \multicolumn{1}{c|}{$S$} &
 	            	   \multicolumn{1}{c||}{$Q$} &
                    		\multicolumn{1}{c|}{\textbf{log($\epsilon$)}}\\ \hline
10M & 1,000 & 10,000 & 25 & 2 & 40,000 & 18,000 &-80.1\\ 
50M & 2,000 & 25,000 & 30 & 2 & 86,000 & 40,000 &-81.8 \\ 
100M & 3,000 & 33,334 & 30 & 2 & 117,000 & 57,000 &-81.9\\ 
200M & 4,400 & 45,455 & 24 & 2 & 170,000 & 73,000 &-64.5\\
\end{tabular}
\caption{Stash Shuffle parameter scenarios and their security.\vspace*{1.5ex}}
\label{tab:stash-shuffle}
\end{table}

\end{document}